\ifx\synctex\undefined\else\synctex=1\fi

\documentclass[a4paper,UKenglish,english,cleveref,autoref]{lipics-v2019}

\hideLIPIcs
\nolinenumbers

\usepackage[utf8]{inputenc}

\usepackage{graphicx}
\usepackage{listings}

\usepackage{latexsym}
\usepackage{hyperref}
\usepackage{mathtools}
\usepackage{subcaption} 
\usepackage{tikz}
\usepackage[noend]{algpseudocode}
\usepackage{algorithm,algorithmicx}
\usetikzlibrary{arrows}
\usepackage{extarrows}

\usepackage[textsize=small]{todonotes}

\usepackage{tikz}
\usepackage{tkz-graph}
\usetikzlibrary{automata,fit}
\usetikzlibrary{positioning}

\tikzstyle{VertexStyle} = [shape            = ellipse,
minimum width    = 6ex,%
draw]
\tikzstyle{EdgeStyle}   = [->,>=stealth']      

\lstdefinelanguage{algo}{%
  morekeywords={function,push,pop,top,for,all,and,or,not,if,then,else,repeat,until,while,do,report,return,such,that,int,stack,end,delete,let,procedure,call,
  remove, skip, add}
}

\usepackage{multirow}
\usepackage{multicol}

\newcommand{\Nat}{\ensuremath{\mathbb{N}}}

\newcommand{\Land}{\bigwedge}

\newcommand{\es}{\emptyset}
\newcommand{\incl}{\subseteq}

\newcommand{\sat}{\vDash}

\newcommand{\fleq}{\preccurlyeq}

\newcommand{\sq}[1]{[#1]}
\newcommand{\di}[1]{\langle #1 \rangle}
\newcommand{\sqq}[1]{\sq{\cdot }}
\newcommand{\ddi}[1]{\di{\cdot }}
\newcommand{\set}[1]{\{#1\}}

\def\sub[#1/#2]{[#1/#2]}

\newcommand{\act}[1]{\stackrel{#1}{\longrightarrow}}

\renewcommand{\a}{\alpha}
\renewcommand{\b}{\beta}
\renewcommand{\d}{\delta}

\newcommand{\s}{\sigma}

\renewcommand{\th}{\theta}

\newcommand{\D}{\Delta}

\renewcommand{\S}{\Sigma}

\newcommand{\Nn}{\mathcal{N}}

\newcommand{\PSPACE}{\text{\sc Pspace}}

\newcommand{\struct}[1]{\langle #1 \rangle}

\def\sqr#1#2{\vbox
 {\hrule height#2
  \mbox{\vrule width#2 height#1 \kern#1 \vrule width#2}%
  \hrule height#2}}

\newcommand{\wt}{\widetilde}

\newcommand{\elapse}[1]{\overrightarrow{#1}}
\newcommand{\sync}{\operatorname{sync}}

\newcommand{\init}{\mathit{init}}
\newcommand{\lelapse}{\operatorname{local-elapse}}
\newcommand{\wx}{\wt x}
\newcommand{\wy}{\wt y}
\newcommand{\wX}{\wt X}

\newcommand{\wz}{\wt z}

\newcommand{\lstep}[1]{\xlongrightarrow{#1}_{st}}
\newcommand{\lact}[1]{\xlongrightarrow{#1}}

\renewcommand{\act}[1]{\xLongrightarrow{#1}}
\newcommand{\step}[1]{\xLongrightarrow{#1}_{st}}

\renewcommand{\sub}{\triangledown}

\newcommand{\lzg}{\mathrm{LZG}}

\newcommand{\dom}{\mathit{dom}}
\newcommand{\Proc}{\mathit{Proc}}

 \newcommand{\lv}{\ensuremath{\mathsf{v}}}
 \newcommand{\lZ}{\ensuremath{\mathsf{Z}}}
 \newcommand{\cv}{\ensuremath{\overline{v}}}

\newcommand{\Rpos}{\mathbb{R}_{\ge 0}}

\renewcommand{\Nat}{\mathbb{N}}

\newcommand{\abs}{\mathfrak{a}}

\newcommand{\lval}{\ensuremath{\mathsf{local}}}
\newcommand{\gval}{\ensuremath{\mathsf{global}}}

\newcommand{\tto}[1]{\xRightarrow{#1}}
\renewcommand{\part}{\operatorname{dom}}

\makeatletter
\newcommand{\oset}[2]{%
  {\mathop{#2}\limits^{\vbox to -.5\ex@{\kern-\tw@\ex@
   \hbox{\scriptsize #1}\vss}}}}
\makeatother

\newcommand{\lleq}{\triangleleft}

\newcommand{\cmax}{c_{\operatorname{max}}}
\newcommand{\mineareg}{\simeq_{\operatorname{reg}}}
\newcommand{\syncincl}{\sqsubseteq^{\abs}_{sync}}

\newcommand{\invariantref}[1]{\hyperref[#1]{#1}}

\title{Revisiting local time semantics for networks of timed automata}

\author{R. Govind}{Chennai Mathematical Institute, India \and
  Univ. Bordeaux, CNRS,  Bordeaux INP, LaBRI, UMR 5800, F-33400,
  Talence, France}{govindr@cmi.ac.in}{}{}

\author{Frédéric Herbreteau}{Univ. Bordeaux, CNRS, Bordeaux INP,
  LaBRI, UMR 5800, F-33400, Talence, France}{fh@labri.fr}{}{}

\author{B. Srivathsan}{Chennai Mathematical Institute,
  India}{sri@cmi.ac.in}{}{}

\author{Igor Walukiewicz}{Univ. Bordeaux, CNRS, Bordeaux INP,
  LaBRI, UMR 5800, F-33400, Talence, France}{igw@labri.fr}{}{}

\authorrunning{R. Govind, F. Herbreteau, B. Srivathsan and
  I. Walukiewicz}

\Copyright{R. Govind, Frédéric Herbreteau, B. Srivathsan and
  Igor Walukiewicz}



\funding{Work supported by UMI 2000 ReLaX and project IoTTTA - CEFIPRA
  Indo-French program in ICST - DST/CNRS ref. 2016-01. Author
  B. Srivathsan is partially funded by grants from Infosys Foundation,
  India and Tata Consultancy Services, India}

\begin{document}

\maketitle

\begin{abstract}
  We investigate a zone based approach for the reachability problem in
  timed automata. 
  The challenge is to alleviate the size explosion of the search space when
  considering networks of timed automata working in parallel.
  In the timed setting this explosion is particularly visible as even
  different interleavings of local actions of processes may lead
  to different zones.
  Salah \emph{et al.}\ in 2006 have shown that the union of all these
  different zones is also a zone. This observation was used in an
  algorithm which from time to time detects and aggregates these zones
  into a single zone.
  
  We show that such aggregated zones can be calculated more efficiently
  using the local time semantics and the related notion of local zones
  proposed by Bengtsson \emph{et al.} in 1998. Next, we point out a flaw
  in the existing method to ensure termination of the local zone graph
  computation. We fix this with a new algorithm that builds the local
  zone graph and uses abstraction techniques over (standard) zones for termination.
  We evaluate
   our algorithm on standard examples. On various 
  examples, we observe an order of magnitude decrease in the 
   search space. On the other examples, the algorithm performs like the
   standard zone algorithm. 
\end{abstract}

\section{Introduction}

Timed automata~\cite{Alur:TCS:1994} are a popular model for real-time
systems.
They extend finite state automata with real valued
variables called \emph{clocks}.
Constraints on clock values can be
used as guards for transitions, and clocks can be reset to zero during
transitions.
Often, it is more natural to use a network of
timed automata which operate concurrently and
synchronize on joint actions.  We study the reachability
problem for networks of timed automata: given a state of the timed
automaton network, is there a run from the initial state to the given
state.

A widely used technique to solve the reachability problem constructs a
\emph{zone graph}~\cite{Daws}
whose nodes are \emph{(state, zone)} pairs consisting of a state
of the automaton and a \emph{zone} representing a set of clock
valuations~\cite{Dill:1990:DBM}.
This graph may not be finite, so in order to guarantee termination
of an exploration algorithm, various sound and complete abstraction
techniques are
used~\cite{Daws,BBLP-STTT05,Herbreteau,Gastin:CAV:2019}.

Dealing with automata operating in parallel poses the usual
state-space explosion problem arising due to different
interleavings. Consider an example of a network with two processes in
Figure~\ref{fig:local-zones-commute}.
\begin{figure}[t]
  \begin{tikzpicture} [state/.style={draw, thick, circle, inner
    sep=2pt},font=\scriptsize]
  \begin{scope}[yshift=-0.5cm]
    \begin{scope}[every node/.style={state}]
      \node (p0) at (0,0) {$p_0$};
      \node (p1) at (0,-1.2) {$p_1$};
      \node (q0) at (1, 0) {$q_0$};
      \node (q1) at (1,-1.2) {$q_1$};
    \end{scope}
    
    \begin{scope}[->, >=stealth, thick]
      \draw ([yshift=0.5cm]p0.north) to (p0);
      \draw (p0) edge node[left]
      {$\begin{array}{c}a\\ \{x\}\end{array}$} (p1);
      \draw ([yshift=0.5cm]q0.north) to (q0);
      \draw (q0) edge node[right]
      {$\begin{array}{c}b\\ \{y\}\end{array}$}(q1);
    \end{scope}

    \begin{scope}[font=\footnotesize]
      \node at (0, -1.7) {$A_1$};
      \node at (1, -1.7) {$A_2$};
    \end{scope}
  \end{scope}
  
  \begin{scope}[xshift=5cm]
    \begin{scope}[zone/.style={draw,rectangle,rounded corners,inner sep=1.5pt}]
      \node [zone, fill=gray!20] (z0) at (0, 0)
      {$\begin{array}{c}t \geq 0\\ \wx=\wy=0\end{array}$};
      \node [zone] (z1) at (-1.5, -1)
      {$\begin{array}{c}t \geq \wx\\ \wx \geq \wy = 0\end{array}$};
      \node [zone] (z2) at (1.5, -1)
      {$\begin{array}{c}t \geq \wy\\ \wy \geq \wx = 0\end{array}$};
      \node [zone] (z3) at (-1.5, -2.1)
      {$\begin{array}{c}t \geq \wy\\ \wy \geq \wx \geq 0\end{array}$};
      \node [zone] (z4) at (1.5, -2.1)
      {$\begin{array}{c}t \geq \wx\\ \wx \geq \wy \geq 0\end{array}$};
    \end{scope}
    
    \begin{scope}[->, >=stealth, thick]
      \draw (z0) edge node[left=2pt] {$a$}  (z1);
      \draw (z0) edge node[right=2pt] {$b$} (z2);
      \draw (z1) edge node[left] {$b$}  (z3);
      \draw (z2) edge node[right] {$a$} (z4);
    \end{scope}
  \end{scope}
  
  \begin{scope}[xshift=10.5cm]
    \begin{scope}[zone/.style={draw,rectangle,rounded corners,inner sep=1.5pt}]
      \node [zone, fill=gray!20] (z0) at (0, 0)
      {$\begin{array}{c} t_1 \geq \wx = 0\\ t_2 \geq \wy = 0 \end{array}$};
      \node [zone] (z1) at (-1.5, -1)
      {$\begin{array}{c} t_1 \geq \wx \geq 0\\ t_2 \geq \wy = 0 \end{array}$};
      \node [zone] (z2) at (1.5, -1)
      {$\begin{array}{c} t_1 \geq \wx = 0\\ t_2 \geq \wy \geq 0 \end{array}$};
      \node [zone] (z3) at (0, -2.1)
      {$\begin{array}{c} t_1 \geq \wx \geq 0\\ t_2 \geq \wy \geq 0\end{array}$};
    \end{scope}
    
    \begin{scope}[->, >=stealth, thick]
      \draw (z0) edge node[left=2pt] {$a$} (z1);
      \draw (z0) edge node[right=2pt] {$b$} (z2);
      \draw (z1) edge node[left=2pt] {$b$} (z3); 
      \draw (z2) edge node[right=2pt] {$a$} (z3);
    \end{scope}
  \end{scope}

  \begin{scope}[font=\footnotesize,yshift=-2.7cm]
    \node at (0.5, 0) {Network $\langle A_1, A_2\rangle$};
    \node at (5, 0) {Global zone graph};
    \node at (10.5, 0) {Local zone graph};
  \end{scope}
\end{tikzpicture}

  \caption{Illustration of commutativity in the local zone graph}
  \label{fig:local-zones-commute}
\end{figure}
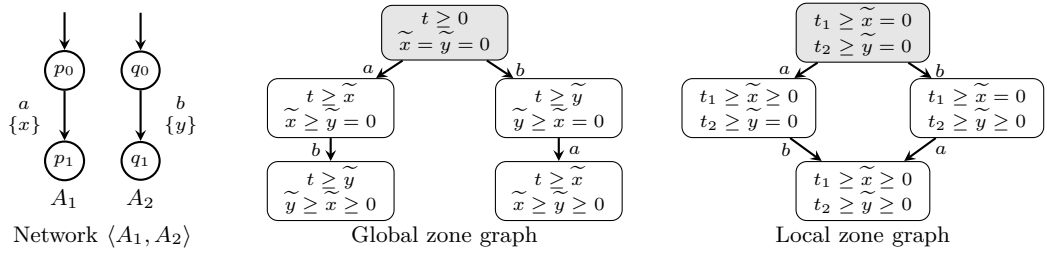
Actions $a$ and $b$ are local to each process. Variables $x, y$ are
clocks and $\{x\}$ denotes that clock $x$ is reset in transition
$a$.
The (global) zone graph maintains the set of configurations reached
after each sequence of actions. Although $a$ and $b$ are local
actions, there is an intrinsic dependence between the two
processes happening due to time. Hence, the zone reached after
executing sequence $ab$ contains configurations where
$x$ is reset before $y$
and the zone after $ba$ contains configurations where
$y$ is reset before $x$.
So the sequences $ab$ and $ba$ lead to different zones.
The number of different interleavings of sequences of local actions increases
exponentially when their length grows, or when more processes get
involved.
Every interleaving can potentially lead to a different zone.

Salah et al. \cite{Salah2006OnII} have shown a surprising property
that the \emph{union} of all zones reached by the interleavings of a
sequence of actions is also a zone. 
We call it an \emph{aggregated zone}.
Their argument is based on the fact that for a given sequence one can
write a zone-like 
constraint defining all the runs on the interleavings of the
sequence.
Then the aggregated zone is obtained simply by projecting this big
constraint on relevant components.
This approach requires to work with sets of constraints whose size
grows with the length of a sequence.
This is both inefficient and limited to finite sequences. 
They use this observation in an algorithm where,
when all the interleavings of a sequence $\s$ have been explored, the
resulting zones are aggregated to a single zone and further
exploration is restricted to this aggregated zone. This
requires
detecting from time to time whether aggregation can
happen.
This is an obstacle in using aggregated zones in efficient
reachability-checking algorithms. Another limitation of this approach
is that it works only for acyclic automata.

Another approach by Bengtsson et al.~\cite{Bengtsson:CONCUR:1998}
involves making time local to each process. This local time approach
is based on a very elegant idea: make time in every process progress
independently, and synchronize local times of processes when they need
to perform a common action.  In consequence, the semantics has the
desired property: two actions whose process domains are disjoint are
commutative. In the example above, depending on the local time in processes $A_1$
and $A_2$, the sequence $ab$ may result, on a global time scale,
in $a$ occurring before $b$, as well as $b$ occurring before $a$. As a
result, the set of valuations reached after $ab$ does not remember the
order in which $a$ and $b$ occurred. Thus, sequences $ab$ and $ba$
lead to the same set of configurations: those obtained after doing $a$
and $b$ concurrently.
Similar to standard zones, we now have
local zones and a local zone graph having \emph{(state, local zone)}
pairs. Due to the above argument, the local zone graph has a nice
property --- if a run $\s$ from an initial zone reaches a local 
zone, then all the runs equivalent to $\s$ lead to the same local
zone. Local zone graphs are therefore ideal for handling interleavings. However,
substantially more involved abstraction techniques are needed for local zones to
make the local zone graph finite. A \emph{subsumption relation}
between local zones was defined in \cite{Bengtsson:CONCUR:1998} but no
algorithm was proposed and there was no effective way to use local
time semantics. Later Minea \cite{Minea} proposed a widening operator
on local zones to construct finite local zone graphs.

\subparagraph*{Summary of results in the paper:}
\begin{itemize}
\item We show that the aggregated zone of interleavings of $\s$ in the
  standard zone graph is obtained by synchronizing valuations in the
  local zone obtained after $\s$
  (Theorem~\ref{thm:local_zones_are_aggregated_zones}). Hence
  computing the local zone graph gives a more direct and efficient
  algorithm than Salah et al. to compute aggregated zones. Moreover,
  this algorithm is not restricted to acyclic timed automata.
  
\item We point out a flaw in the abstraction procedure of Minea to get
  a finite local zone graph
  (Section~\ref{sec:making-offset-zone}).
  
\item We propose a different algorithm to get a finite local zone
  graph. This gives a new reachability algorithm for networks of timed
  automata, which works with local zones but uses subsumption on
  standard zones (Definition~\ref{def:sync-graph} and
  Theorem~\ref{thm:main}). Instead of subsumptions between local
  zones, we use subsumptions between the synchronized valuations
  inside these local
  zones. This helps us to exploit the (well-studied) subsumptions over
  standard zones~\cite{Daws,BBLP-STTT05,Herbreteau,Gastin:CAV:2019}. 
  Moreover, this subsumption is much more aggressive than the standard one
  since, thanks to local-time semantics, the (aggregated) zone used in the
  subsumption represents all the valuations reachable not only by the
  execution that we are exploring but also by all the executions
  equivalent to it.

\item We report on experiments performed with a prototype
  implementation of the algorithm (Table
  \ref{table:experimental-results}).
  The algorithm performs surprisingly well on some examples, and it is
  never worse than the standard zone graph algorithm.
\end{itemize}

\subparagraph*{Related work:}
The basis of this work is local-time semantics and local zones
developed by Bengtsson et.\ al.~\cite{Bengtsson:CONCUR:1998}. The
authors have left open how to use this semantics to effectively
compute the local zone graph since no efficient procedure was provided
to ensure its finiteness.
To that purpose, Minea~\cite{Minea:thesis,Minea} has proposed a
subsumption operation on local zones, and an algorithm using this
operation. Unfortunately, as we exhibit here, the algorithm has a flaw
that is not evident to repair.
Lugiez et. al.~\cite{DBLP:journals/tcs/LugiezNZ05} also use local time
but their method is different.  They use constraints to check if local
clocks of an execution can be synchronized sufficiently to obtain a
standard run.

Aggregated zones are crucial to obtain an efficient verification
procedure for networks of timed automata. Coming to the same state
with different zones inflicts a huge blowup in the zone graph, since
the same paths are explored independently from each of these
zones. This has also been observed in the context of multi-threaded
program verification in~\cite{Sousa:CAV:2017}. Solving this problem in
the context of program analysis requires to over approximate the
aggregated state. Fortunately, in the context of timed automata, the
result of aggregating these zones in still a
zone~\cite{Salah2006OnII}.  Efficient computation of aggregated zones
is thus an important advance in timed automata verification as
demonstrated by our experimental results in
Section~\ref{sec:experiments}.


\section{Networks of timed automata}
\label{sec:preliminaries}

We start by defining networks of timed automata and two semantics for them:
a global-time semantics (the usual one) and a local-time semantics
(introduced in~\cite{Bengtsson:CONCUR:1998}). Then, we recall the fact that
they are equivalent, and state some interesting properties of
local-time semantics w.r.t. concurrency that were observed
in~\cite{Bengtsson:CONCUR:1998}. 

Let $\Nat$ denote the set of natural numbers and $\Rpos$ the set of
non-negative reals. Let $X$ be a finite set of variables called
\emph{clocks}. Let $\phi(X)$ denote a set of clock constraints
generated by the following grammar:
$\phi := x \sim c ~|~ \phi \land \phi$ where $x \in X$, $c \in \Nat$,
and $\sim ~\in \{<, \le, =, \ge, > \}$.

\begin{definition}[Network of timed automata]\label{df:network}
  A network of timed automata with $k$ processes, is a $k$-tuple of
  timed automata $A_1,\dots,A_k$.  Each process
  $A_p=\struct{Q_p,\S_p,X_p,q^\init_p,T_p}$ has a finite set of states
  $Q_p$, a finite alphabet of actions $\S_p$, a finite set of clocks
  $X_p$, an initial state $q^\init_p$, and transitions
  $T_p\incl\S_p\times Q_p\times \phi(X_p)\times 2^{X_p}\times Q_p$.  We
  require that the sets of states, and the sets of clocks are pairwise
  disjoint: $Q_{p_1}\cap Q_{p_2}=\es$, and $X_{p_1}\cap X_{p_2}=\es$
  for $p_1\not=p_2$. We write $\Proc$ for the set of all processes.
\end{definition}

A network is a parallel composition of timed automata.  Its semantics is
that of the timed automaton obtained as the ``synchronized product''
of the processes.  
For an action $b$, a $b$-transition of a process $p$ is an element of $T_p$
with $b$ in the first component.
Synchronization happens on two levels: (i) via time
that advances the same way in all the processes, and (ii) via common
actions, for example if $b\in \S_1\cap \S_2$, then processes $1$ and
$2$ need to synchronize by doing a $b$-transition. 
We define the domain of
an action $b$: $\dom(b)=\set{p : b\in \S_p}$ as the set of processes
that must synchronize to do $b$. We will use some abbreviations:
$Q=\Pi_{p=1}^k Q_p$, $\S=\bigcup_{p=1}^k \S_p$ and
$X=\bigcup_{p=1}^k X_p$.

The semantics of a network is governed by the value of clocks at each
instant. We choose to represent these values using offsets as this
allows a uniform presentation of the global-time semantics below and
the local-time semantics in Section~\ref{sec:local_time_semantics}.
For every clock $x$, we introduce an offset variable $\wx$. The value
of $\wx$ is the time-stamp at which $x$ was last reset. In addition, we
consider a variable $t$ which tracks the global time: essentially $t$
is a clock that is never reset.
We now make this notion precise. Let
$\wX_p = \{ \wx \mid x \in X_p \}$ and $\wX = \bigcup_{p=1}^k \wX_p$.
A \emph{global valuation} $v$ is a function
$v: \wX \cup \{t\} \mapsto \Rpos$ such that $v(t) \ge v(\wx)$ for all
variables $\wx$. In this representation, the value of clock $x$
corresponds to $v(t)-v(\wx)$, denoted $\cv(x)$\footnote{Usually,
  semantics of timed automata is described using valuations of the
  form $\cv$. Here, we have chosen $v$ which uses offsets since it
  extends naturally to the local-time setting. Translations between
  these two kinds of valuations is straightforward, as shown.}.

Recall that offset variable $\wx$ stores the last time-stamp at which
clock $x$ has been reset. Hence, a delay in offset representation
increases the value of reference clock $t$ and leaves offset variables
$\wx$ unchanged. Formally: for $\d \in \Rpos$, we denote by $v + \d$
the valuation defined by: $(v + \d)(t) = v(t) + \d$ and
$(v + \d)(\wx) = v(\wx)$ for all $\wx$. Similarly, resetting the
clocks in $R \subseteq X$, yields a global valuation $[R]v$ defined
by: $([R]v)(t) = v(t)$, and $([R]v)(\wx)$ is $v(t)$ if $x \in R$, and
$v(\wx)$ otherwise. Given a global valuation $v$ and a clock
constraint $g$, we write $v \models g$ if every constraint in $g$
holds after replacing $x$ with its value $\cv(x) = v(t) - v(\wx)$.

A configuration of the network is a pair $(q,v)$ where $q\in Q$ is a
global state, and $v$ is a global valuation.  We will write $q(p)$ to
refer to the $p$-th component of the state $q$.

\begin{definition}[Global-time semantics]\label{def:global-run}
  The semantics of a network $\Nn$ is given by a transition system
  whose states are configurations $(q, v)$. The initial configuration
  is $(q^\init, v^\init)$ where $q^\init(p)=q_p^\init$ is the tuple of
  initial states, and $v^\init(y)=0$ for $y \in \wX \cup \{t\}$.

  There are two kinds of
  transitions, which we call steps: global delay, and action steps.  A
  \emph{global delay} by the amount $\d\in\Rpos$ gives a step
  $(q,v)\step{\d}(q,v+\d)$.  An \emph{action step} on action $b$ gives
  $(q,v)\step{b}(q',v')$ if there is a set of $b$-transitions
  $\set{(q_p,g_p,R_p,q'_p)}_{p\in\dom(b)}$ of the respective processes
   such that:
  \begin{itemize}
    \item  processes from $\dom(b)$ change states: $q_p=q(p)$, $q'_p=q'(p)$, for
    $p\in\dom(b)$, and $q(p)=q'(p)$  for   $p\not\in\dom(b)$; 
    \item all the guards are satisfied: $v\sat g_p$, for $p\in\dom(b)$;
    \item all resets are performed: $v'=[\bigcup_{p\in\dom(b)} R_p]v$;
  \end{itemize}
\end{definition}

A \emph{run of an automaton} from a configuration $(q_0,v_0)$ is a
sequence of steps starting in $(q_0,v_0)$.  For a sequence
$u=b_1\dots b_n$ of actions, we write $(q_0,v_0)\act{u}(q_n,v_n')$ if
there is a run
\begin{equation*}
  (q_0,v_0)\step{\d_0}
  (q_0,v'_0)\step{b_1}
  (q_1,v_1)\step{\d_1}
  (q_1,v'_1)\dots\step{b_n}
  (q_n,v_n)\step{\d_n}
  (q_n,v'_n)
\end{equation*}
for some delays  $\d_0,\dots,\d_n \in \Rpos$.

\begin{definition}[Reachability problem]
  The \emph{reachability problem } is to decide, given a network $\Nn$
  and a state $q$, whether there is a run reaching $q$; or in other
  words, whether there exists a sequence of transitions $u$ such that $(q^\init,v^\init) \act{u} (q,v)$ for some valuation $v$.
\end{definition}

The reachability problem for networks of timed automata is
\PSPACE-complete~\cite{Alur:TCS:1994}.

\subsection{Local-time semantics}
\label{sec:local_time_semantics}

The definition of a network of automata suggests an independence
relation between actions: a pair of actions with disjoint domains
(i.e. involving distinct processes) should commute. We say that
two sequences of actions are \emph{equivalent}, written $u\sim w$ if
one can be obtained from the other by repeatedly permuting adjacent
actions with disjoint domains.

\begin{lemma}\label{lem:small-commutation}
  For two equivalent sequences $u\sim w$: if there are two runs
  $(q,v)\act{u}(q_{u},v_{u})$, and $(q,v)\act{w}(q_{w},v_{w})$ then
  $q_{u}=q_{w}$.
\end{lemma}
\begin{proof}
  Consider the basic case when there are two actions $a$ and $b$ with
  disjoint domains and $u = ab$, $w = ba$. Since $a$ and $b$ are on
  disjoint processes, executing $ab$ or $ba$ (whenever possible) leads
  to the same (discrete) state by definition.

  Consider a general $u$. Sequence $w$ is obtained by repeatedly
  permuting adjacent actions. From the basic case of the lemma, each
  permutation preserves the source and target (discrete) states, if it
  is feasible.
\end{proof}

Observe that in the above lemma we cannot assert that $v_{u}=v_{w}$.
Even further, the existence of $(q,v) \act{u} (q_{u},v_{u})$ does not
imply that a run from $(q,v)$ on $w$ is feasible. This happens due to
global time delays, i.e., delays that involve all the processes. For
example, consider actions $a$ and $b$ on disjoint
processes with $a$ having guard $x \le 1$ and $b$ having guard
$y \ge 2$. Then from the initial valuation one can execute $ab$ but
not $ba$.

One solution to get commutativity between actions with disjoint
domains is to consider local-time
semantics~\cite{Bengtsson:CONCUR:1998}. In this semantics, time elapses
independently in every process, and time elapse is synchronized before
executing a synchronized action.  This way, two actions with disjoint
domains become commutative. In the example from the previous
paragraph, while the process 
executing $b$ elapses $2$ time units, the other process is allowed to
not elapse time at all and hence $ba$ becomes possible.  Moreover, for
the reachability problem, local-time semantics is equivalent to the
standard one. 

In the local-time semantics, we replace the clock $t$ which was
tracking the global time, with individual \emph{reference clocks
  $t_p$} for each process $A_p$ which track the local time of each
process.  We set $\wX_p'= \wX_p \cup \set{t_p}$ and
$\wX' = \bigcup_p \wX'_p$.  A \emph{local valuation} $\lv$ is a
valuation over the set of clocks $\wX'$ such that
$\lv(t_p)\geq \lv(\wx)$ for all processes $p\in \Proc$ and all clocks
$\wx\in \wX_p$\label{page:local-valuation-restriction}.  This
restriction captures the intuition that $t_p$ is a reference clock for
process $p$, and it is never reset. In this setting, the value
$\lv(t_p)-\lv(\wx)$ of clock $x$ is defined relative to the
reference clock $t_p$ of process $p$ that owns $x$, i.e. $x \in X_p$.
We will use the notation $\lv$ for local valuations to distinguish
from global valuations $v$.

We introduce a new operation of local time elapse.  For a process
$p\in\Proc$ and $\d\in\Rpos$, operation $\lv+_p\d$ adds $\d$ to
$\lv(t_p)$, the value of the reference clock $t_p$ of process $p$, and
leaves the other variables unchanged. Formally,
$(\lv+_p\d)(t_p) = \lv(t_p) + \d$ and $(\lv +_p \d)(y) = \lv(y)$ for
all $y \in \wX' \setminus \{t_p\}$. A local valuation $\lv$ satisfies
a clock constraint $g$, denoted $\lv \models g$ if every constraint in
$g$ holds after replacing $x$ by its value $\lv(t_p) - \lv(\wx)$ where
$p$ is the process such that $x \in X_p$. We denote by $[R]\lv$ the
valuation obtained after resetting the clocks in $R \incl X$ and
defined by: $([R]\lv)(t_p) = \lv(t_p)$ for every reference clock
$t_p$, $([R]\lv)(\wx) = \lv(\wx)$ if $x \not\in R$, and
$([R]\lv)(\wx) = \lv(t_p)$ if $x \in R$ and $p$ is the process such
that $x \in X_p$.

\begin{definition}[Local steps of a timed automata network]
  There are two kinds of local steps in a network $\Nn$: \emph{local
    delay}, and \emph{local action}.  A local delay $\d \in \Rpos$ in
  process $p \in Proc$ is a step $(q,\lv)\lstep{p,\d}(q,\lv+_p\d)$.
  For an action $b$, we have a step $(q,\lv)\lstep{b}(q',\lv')$ 
  if there is a set of $b$-transitions of respective processes
  $\set{(q_p,g_p,R_p,q'_p)}_{p\in\dom(b)}$
   such that:
  \begin{itemize}
    \item  $q_p=q(p)$, $q'_p=q'(p)$, for $p\in\dom(b)$, and $q(p)=q'(p)$  for
    $p\not\in\dom(b)$; 
    \item start times are synchronized: $\lv(t_{p_1})=\lv(t_{p_2})$, for
    every $p_1,p_2\in\dom(b)$;
    \item guards are satisfied: $\lv\sat g_p$, for every $p\in\dom(b)$;
    \item resets are performed: $\lv'=[\bigcup_{p\in \dom(b)}R_p]\lv$;
  \end{itemize}
\end{definition}

The main difference with respect to global semantics is the presence
of local time delay.  As a result, every process can be in a different
local time as emphasized by the reference clocks in each process.  In
consequence, in local action steps we require that when processes do a
common action their local times should be the same.
Of course a standard delay $\d$ on all processes can be simulated by a
sequence of delays on every process separately, as
$\lstep{1,\d}\cdots\lstep{k,\d}$.  For a sequence of local delays
$\D=(p_1,\d_1)\dots(p_n,\d_n)$ we will write
$(q,\lv)\lstep{\D}(q,\lv')$ to mean
$(q,\lv)\lstep{p_1,\d_1}(q,\lv_1)\lstep{p_2,\d_2}\cdots\lstep{(p_n,\d_n)}(q,\lv')$.

\begin{definition}[Local run]\label{def:local-run}
  A \emph{local run} from a configuration $(q_0,\lv_0)$ is a sequence
  of local steps.  For a sequence of actions $u=b_1\dots b_n$, we
  write $(q_0,\lv_0)\lact{u}(q_n,\lv_n')$ if for some sequences of local
  delays $\D_0,\dots,\D_n$ there is a local run
  \begin{equation*}
    (q_0,\lv_0)\lstep{\D_0}
    (q_0,\lv_0')\lstep{b_1}
    (q_1,\lv_1) \lstep{\D_1}
    \cdots
    \lstep{b_n}
    (q_n,\lv_n)\lstep{\D_n}(q_n,\lv'_n)
  \end{equation*}
\end{definition}
Observe that a run may start and end with a sequence of delays.  In
the next section we will make a link between local and global runs.
For this we will first examine independence properties of local runs
which are much better than for global runs (cf.\
Lemma~\ref{lem:small-commutation}).

\begin{lemma}[Independence]\label{lem:independence}
  Suppose $\dom(a)\cap\dom(b)=\es$. If $(q,\lv)\lact{ab}(q',\lv')$
  then $(q,\lv)\lact{ba}(q',\lv')$.  If $(q,\lv)\lact{a}(q_a,\lv_a)$
  and $(q,\lv)\lact{b}(q_b,\lv_b)$ then
  $(q,\lv)\lact{ab}(q_{ab},\lv_{ab})$ for some $q_{ab}$ and
  $\lv_{ab}$.
\end{lemma}
\begin{proof}
 Take a run
  $(q,\lv)\lact{\D_a}(q,\lv_a)\lact{a}(q_a,\lv'_a)\lact{\D_b}(q_a,\lv_b)\lact{b}(q',\lv_b')
  \lact{\D} (q', \lv')$.
  Let $\D'_b$ be a sequence of delays from $\D_a$ or $\D_b$
  involving a process $p\in \dom(b)$, i.e., pairs $(p,\d_p)$ from
  $\D_a$ or $\D_b$ such that $p \in \dom(b)$. Let $\D'_a$ be a
  sequence of delays in $\D_a$ involving processes in $\dom(a)$; and
  finally let $\D'$ be the delays in $\D_a \cup \D_b \cup \D$ which
  were not counted in $\D'_a$ or $\D'_b$.   
  Since $\dom(a)\cap\dom(b)=\es$ we get
  that the following sequence is a run:
  $(q,\lv)\lact{\D'_b}(q,\lv''_b)\lact{b}(q_b,\lv'''_b)\lact{\D'_a}(q',\lv''_a)\lact{a}(q',\lv'''_a)\lact{\D'}(q',\lv')$.

  For the second part, suppose $(q, \lv) \lact{\D_a} \lact{a} \lact{\D'_a}
  (q_a, \lv_a)$ and $(q, \lv)  \lact{\D_b} \lact{b} \lact{\D'_b}  (q_b,
  \lv_b)$. Let $\D_1$ be the delays involving processes in $\dom(a)$ in
  $\D_a$, and $\D_2$ the delays of processes in $\dom(b)$ in
  $\D_b$. As $\dom(a) \cap \dom(b) = \emptyset$, $(q,\lv) \lact{\D_1} \lact{a}\lact{\D_2}\lact{b} (q_{ab},
  \lv_{ab})$ is a  local run.
\end{proof}

Recall that two sequences of actions are equivalent, written $u\sim w$
if one can be obtained from the other by repeatedly permuting adjacent
actions with disjoint domains.  Directly from the previous lemma we
obtain.

\begin{lemma}\label{lem:independence-in-local-valuations}
  If $(q_0,\lv_0)\lact{u}(q_n,\lv_n)$ and $u\sim w$ then
  $(q_0,\lv_0)\lact{w}(q_n,\lv_n)$.
\end{lemma}

With local-time semantics two equivalent sequences not only
reach the same state $q_n$, but also
the same local valuation $\lv_n$ (in contrast with
Lemma~\ref{lem:small-commutation} for global-time semantics).

\subsection{Comparing local and global runs}

We have presented two semantics for networks of timed automata:
global-time and local-time. Local runs have more freedom as time can elapse
independently in every process. Yet, with respect to state reachability
the two concepts turn out to be  equivalent.

\begin{definition}\label{def:synchronized-val}
  A local valuation $\lv$ is \emph{synchronized} if for every
pair of processes $p_1,p_2$, the values of their reference clocks are equal:
$\lv(t_{p_1})=\lv(t_{p_2})$.
\end{definition}

For a synchronized local valuation $\lv$, let $\gval(\lv)$ be the
global valuation $v$ such that $v(\wx) = \lv(\wx)$ and
$v(t) = \lv(t_1) = \cdots = \lv(t_k)$. Conversely, to every
global valuation $v$, we associate the synchronized local valuation
$\lval(v) = \lv$ where $\lv(\wx) = v(\wx)$ and $\lv(t_p) =
v(t)$ for every reference clock $t_p$.

Before we prove the main observation of this section (Lemma~\ref{lem:eq-local-global}), we develop some
intermediate observations.

Every action step in a local run:

\begin{equation*}
  (q_0,\lv_0)\lstep{\D_0}(q_0,\lv_0')\lstep{b_1}(q_1,\lv_1)
  \dots\lstep{\D_{n-1}}(q_{n-1},\lv'_{n-1})\lstep{b_n}(q_n,\lv_n)\lstep{\D_n}(q_n,\lv'_n)
\end{equation*}
has its execution time; namely the step
$(q_{i-1},\lv'_{i-1})\lstep{b_{i}}(q_{i},\lv_{i})$ has the execution
time $\lv'_{i-1}(t_p)=\lv_{i}(t_p)$ for $p\in\dom(b_{i})$.  Observe
that by definition of a step, the choice of $p$ does not matter.

We will say that a local run is \emph{soon} if for every $i$, the
execution time of $b_i$ is not bigger than the execution time of
$b_{i+1}$.

\begin{lemma}\label{lem:local-means-local-soon}
  If $(q_0,\lv_0)\lact{u}(q_n,\lv_n)$ is a local run then there is
  $w\sim u$ such that $(q_0,\lv_0)\lact{w}(q_n,\lv_n)$ is a soon local
  run.
\end{lemma}
\begin{proof}
  Consider a sequence $u=b_1\dots b_n$ and a run
  $(q_0,\lv_0)\lstep{\D_0}(q_0,\lv_0')\lstep{b_1}(q_1,\lv_1) \cdots$
  $\lstep{\D_{n-1}}(q_{n-1},\lv'_{n-1})\lstep{b_n}(q_n,\lv_n)
  \lstep{\D_{n}} (q_n, \lv'_n)$.  Suppose that the order of execution
  times of $b_i$ and $b_{i+1}$ is reversed.  Then
  $\dom(b_i)\cap\dom(b_{i+1})=\es$ by the definition of a run.  So we
  can take $u'=b_1\dots b_{i+1}b_i\dots b_n$ where the order of $b_i$
  and $b_{i+1}$ is reversed.  Since $u\sim u'$, by Lemma
  \ref{lem:independence-in-local-valuations} we have a run
  $(q_0,\lv_0)\lact{u'}(q_n,\lv_n)$.  We can, if necessary, repeat
  this operation from $u'$ till we get the desired $w$.
\end{proof}

\begin{lemma}\label{lem:run-local-global}
  If $(q,\lv)\lact{u}(q',\lv')$ is a local run where $\lv$ and $\lv'$ are
  synchronized valuations, then there is $w\sim u$ and a global run
  $(q,\gval(\lv))\act{w}(q',\gval(\lv'))$.
\end{lemma}
\begin{proof}
  Let $(q,\lv) = (q_0, \lv_0)$ and $(q',\lv') = (q_n, \lv'_n)$. We take a
  local run, and assume that it is soon thanks to Lemma~\ref{lem:local-means-local-soon}:
  \begin{align*}
    & (q_0,\lv_0)\lstep{\D_0}(q_0,\lv_0')\lstep{b_1}(q_1,\lv_1) \lstep{\D_1} \cdots \\
    & \cdots~ (q_{n-1},\lv_{n-1})\lstep{\D_{n-1}}(q_{n-1},\lv'_{n-1})\lstep{b_n}(q_n,\lv_n)\lstep{\D_n}(q_n,\lv'_n)\ .
  \end{align*}

  Let $\th_i$ be the execution time of action $b_i$.  For convenience,
  we set $\th_0 = \lv_0(t_p)$ and $\th_{n+1}=\lv'_n(t_p)$, for some process $p$.  Since
  $\lv_0$ and $\lv'_n$ is synchronized, the choice of $p$ is irrelevant.  We claim
  that there is a global run
  \begin{equation*}
    (q_0,v_0)\step{\d_1}\step{b_1}(q_1,v_1)\step{\d_2}\step{b_2}
    \cdots\step{\d_i}\step{b_i}(q_i,v_i)    
  \end{equation*}
  with
  \begin{itemize}
  \item $v_0 = \gval(\lv_0)$
  \item $\d_i=\th_{i}-\th_{i-1}$ for
    $i=1,\dots,n$ 
  \item $v_i(t) = \lv_i(t_p)$ for some $p \in \dom(b_i)$ and 
  \item $v_i(\wx) = \lv_i(\wx)$ for all other clocks $x$
  \end{itemize}
  
  This statement is proved by induction on $i$.  For $i=n$, this
  statement gives a global valuation $v_n$ such that $v_n(t) =
  \lv_n(t_p)$ where $p \in \dom(b_n)$ and for all other clocks
  $v_n(\wx) = \lv_n(\wx)$. Note that valuations $\lv_n$ and $\lv'_n$
  differ only in the values of the reference clocks. Moreover, in
  $\lv'_n$, all reference clocks are at $\th_{n+1}$. A global delay of
  $\d_{n+1} = \th_{n+1} - \th_{n}$ from $(q_n, v_n)$ gives $(q_n,
  v'_n)$ such that $v'_n = \gval(\lv'_n)$. 
\end{proof}

\begin{lemma}\label{lem:run-global-local}
If $(q,v)\act{u}(q',v')$ is a global run, then there is a local
  run $(q,\lval(v))\lact{u}(q',\lval(v'))$.
\end{lemma}
\begin{proof}
  A global run can be directly converted to a local run by changing a
  global delay to a sequence of local delays.
\end{proof}

\begin{lemma}\label{lem:eq-local-global}
  If $(q,\lv)\lact{u}(q',\lv')$ is a local run where $\lv$ and $\lv'$
  are synchronized local valuations, there exists a global run
  $(q,\gval(\lv))\act{w}(q',\gval(\lv'))$ for some $w \sim
  u$. Conversely, if $(q,v)\act{u}(q',v')$ is a global run, then there
  is a local run $(q,\lval(v)) \lact{u} (q',\lval(v'))$.
\end{lemma}
\begin{proof}
Follows from lemma \ref{lem:run-local-global}
and \ref{lem:run-global-local}.
\end{proof}

The reachability problem with respect to local semantics is defined as
before: $q$ is reachable if there is a local run
$(q^\init,\lv^\init)\lact{u}(q,\lv)$ for some $\lv$ where
$\lv^\init = \lval(v^\init)$.  By adding some local delays at the end
of the run we can always assume that $\lv$ is synchronized.
Lemma~\ref{lem:eq-local-global} thus implies that the reachability
problem in local semantics is equivalent to the standard one in global
semantics.


\section{Zone graphs}
\label{sec:zone-graphs}

We introduce zones, a standard approach for solving reachability in timed automata. Zones are sets of valuations that can be represented efficiently using simple constraints.

Let us fix a network $\Nn$ of timed automata with $k$
processes. Recall that each process $p$ has a set of clocks $X_p$ and
corresponding offset variables $\wX_p$. The set of clocks in $\Nn$ is
$X = \bigcup_{i=1}^k X_p$. Similarly, the set of offset variables in
$\Nn$ is $\wX = \bigcup_{i=1}^k \wX_p$.

\subsection{Standard zone-based algorithm for reachability}

Recall that in the global-time semantics, $\Nn$ has a reference clock
$t$.  A \emph{global zone} is a set of global valuations described by
a conjunction of constraints of the form $y_1 - y_2 \lleq c$ where
$y_1,y_2 \in \wX \cup \{t\}$, $\lleq\ \in \{<, \leq \}$ and
$c \in \mathbb{Z}$. Since global valuations need to satisfy
$v(\wx) \leq v(t)$ for every offset variable $\wx \in \wX$, a global
zone satisfies $\wx \leq t$ for every $\wx \in \wX$.

Let $g$ be a guard and $R$ a set of clocks. We define the following
operations on zones: $Z_g = \{v \mid v \models g \}$ is the set of
global valuations satisfying $g$, $[R]Z := \{[R]v \mid v \in Z\}$ and
$\elapse{Z} := \{v \mid \exists v' \in Z,~\exists \d \in \Rpos \text{
  s. t. } v = v' + \d\}$. From~\cite{Bengtsson:CONCUR:1998}, $Z_g$,
$[R]Z$ and $\elapse{Z}$ are all zones. We say that a zone is
\emph{time-elapsed} if $Z=\elapse{Z}$.

The semantics of a network of timed automata can be described in terms
of global zones. For an action $b$, consider a set of
$b$-transitions of respective processes
$\set{(q_p,g_p,R_p,q'_p)}_{p\in\dom(b)}$.  Let
$R=\bigcup_{p\in\dom(b)} R_p$, and $g=\Land_{p\in \dom(b)} g_p$.  Then
we have a transition $(q,Z)\act{b} (q',Z')$ where
$Z' = \elapse{[R](Z \cap Z_g)}$ provided that $q(p)=q_p$, $q'(p)=q'_p$ if
$p\in\dom(b)$, and $q'(p)=q(p)$ otherwise, and $Z'$ is not empty. We
write $\tto{}$ for the union over all $\act{b}$.
  
The \emph{global zone graph} $ZG(\Nn)$ of a timed automaton network
$\Nn$ is a transition system whose nodes are of the form $(q, Z)$
where $q \in Q$ and $Z$ is a time-elapsed global zone. The transition
relation is given by $\tto{}$. The initial node is
$(q^\init, Z^\init)$ where $q^\init$ is the tuple of initial states
and $Z^\init = \elapse{\{v^\init\}}$ where $v^{init}$ is the initial
global valuation.
The zone graph $ZG(\Nn)$ is known to be sound and complete with
respect to reachability.  This means that a state $q$ is reachable by
a run of $\Nn$ iff a node $(q,Z)$ for some non-empty $Z$ is reachable
from $(q^\init, Z^\init)$ in the zone graph.  As zone graphs may be
infinite, an abstraction operator is used to obtain a finite quotient.

An \emph{abstraction operator} $\abs : P(\Rpos^X) \rightarrow
P(\Rpos^X)$ \cite{BBLP-STTT05} is a function from sets of valuations to sets of
valuations such that $W \subseteq \abs(W)$ and $\abs(\abs(W)) =
\abs(W)$. Simulation relations between valuations are a convenient way
to construct abstraction operators that are correct for reachability.
A \emph{time-abstract simulation} is a relation between valuations
that depends on a given network $\Nn$.  We say that $v_1$ can be
simulated by $v_2$, denoted $v_1\fleq v_2$ if for every state $q$ of
$\Nn$, and every delay-action step $(q,v_1)\act{\d_1}\act{b}(q',v'_1)$
there is a delay $\d_2\in \Rpos$ such that
$(q,v_2)\act{\d_2}\act{b}(q',v'_2)$ and $v'_1 \fleq v_2'$. The
simulation relation can be lifted to global zones: we say that $Z$ is
simulated by $Z'$, written as $Z \fleq Z'$ if for all $v \in Z$ there
exists a $v' \in Z'$ such that $v \fleq v'$. An abstraction $\abs$
based on $\fleq$ is defined as
$\abs(W) = \{ v~|~\exists v' \in W \text{ with } v \fleq v' \}$. The
abstraction $\abs$ is finite if its range is finite.  Given two nodes
$(q, Z)$ and $(q', Z')$ of $ZG(\Nn)$, $(q, Z)$ is \emph{subsumed} by
$(q', Z')$, denoted $(q, Z) \sqsubseteq^\abs (q', Z')$, if $q = q'$ and
$Z \incl \abs(Z')$.

\begin{remark}
  \label{remark:zones}
  Our definition of zones slightly differs from the standard
  definition in the literature (e.g.~\cite{Bengtsson:LCPN:2003}) since
  we use offset variables to represent clock valuations. Yet, finite
  time-abstract simulations from the
  literature~\cite{BBLP-STTT05,Herbreteau} can be
  adapted to our settings as a simulation over standard valuations
  $\cv$ can be expressed as a simulation over global valuations $v$
  since for every clock $x$, $\cv(x) = v(t) - v(\wx)$, and zones over
  valuations $v$ can be translated to zones over standard valuations
  $\cv$. 
\end{remark}

A finite abstraction $\abs$ allows to construct a finite \emph{global
  zone graph with subsumption} \label{def:zone-graph-subsumption}for a
network of timed automata $\Nn$.
The construction starts from the
initial node of $ZG(\Nn)$.  Using, say, a breadth-first-search (BFS),
for every constructed node we examine all its successors in $ZG(\Nn)$,
and keep only those that are maximal w.r.t. to $\sqsubseteq^\abs$
relation. Computing such a zone graph with subsumption gives an algorithm for
the reachability problem. However, the global zone graph, and hence the
algorithm above, are sensitive to the combinatorial explosion arising
from parallel composition. Global time makes any two actions
potentially dependent - the same is still true on the level of zones.
Zone graphs based on the local-time semantics, as presented next,
solve this problem.

\subsection{Local zone graphs}
\label{sec:local-zones}

The goal of this section is to introduce a concept similar to global
zones and global zone graphs for local-time semantics. Recall that in
the local-time semantics, each process $p$ has a reference clock
$t_p$.

A \emph{local zone} is a zone over local valuations: a set of local
valuations defined by constraints $y_1 - y_2 \lleq c$ where
$y_1, y_2 \in \wX \cup \{t_1, \dots t_k\}$. Recall that a local
valuation $\lv$ needs to satisfy: $\lv(\wx)\leq \lv(t_p)$ for every
process $p$ and every $\wx\in \wX_p$.  This means that a local zone
satisfies $\wx\leq t_p$ for every process $p$ and every
$\wx \in \wX_p$.  We will use $\lZ$, eventually with subscripts, to
range over local zones, and distinguish from global zones $Z$.  Local
zones are closed under all basic operations involved in a local step
of a network of timed automata, namely: local time elapse,
intersection with a guard, and reset of
clocks~\cite{Bengtsson:CONCUR:1998}. That is, for every local zone
$\lZ$:
\begin{itemize}
\item the set
  $\lelapse(\lZ)=\set{\lv+_1\d_1+_2\dots+_k\d_k \mid \lv\in\lZ,\
    \d_1,\dots,\d_k\in\Rpos}$ is a local zone.
\item for every guard $g$ the set $\lZ_g=\set{\lv \mid \lv \sat g}$ is
  a local zone.
\item for every set of clocks $R \incl X$, the set
  $[R]\lZ=\set{[R]\lv \mid \lv\in \lZ}$ is a local zone.
\end{itemize}
Local zones can be implemented using DBMs. Hence, they can be computed
and stored as efficiently as standard zones. The proofs of the above
three statements are given in the subsequent lemmas. 

A zone $\lZ$ is said to be in \emph{canonical form} if every
constraint defining $\lZ$ is tight: in other words, removing some
constraint $y_1 - y_2 \lleq c$ from $\lZ$ results in a different set
of valuations.

\begin{lemma}\label{local-time-elapse}
  The local zone for $\lelapse(\lZ)$ is obtained by removing
  constraints of the form $t_i - \wx \lleq c$, and $t_i - t_j \lleq c$
  for all $i, j$ and $\wx \in \wX$ from the canonical representation
  of $\lZ$.
\end{lemma}
\begin{proof}
  Let $\lZ_1$ be the set of constraints of the form
  $\wx - t_i \lleq c$ and $\wx - \wy \lleq c$ of $\lZ$. To show
  $\lZ_1$ describes $\lelapse(\lZ)$.
  
  $\lelapse(\lZ) \incl \lZ_1$: Each valuation $\lv' \in \lelapse(\lZ)$
  is of the form $\lv +_1 \d_1 +_2 \d_2 \cdots +_k \d_k$ for some
  $\lv \in \lZ$. Hence $\lv'$ is obtained by increasing the values of
  reference clocks from $\lv$ and keeping the other offsets the
  same. This gives $\lv'(\wx) - \lv'(\wy) = \lv(\wx) - \lv(\wy)$ and
  $\lv'(\wx) - \lv'(t_i) \le \lv(\wx) - \lv(t_i)$ for all component
  clocks $\wx, \wy$ and reference clocks $t_i$. Since $\lv \in \lZ$
  and hence satisfies the constraints in $\lZ_1$, valuation $\lv'$
  satisfies them as well.

  $\lZ_1 \incl \lelapse(\lZ)$: Pick $\lv' \in \lZ_1$. Consider a set
  of constraints $\lZ_2$: $\wx = \lv'(\wx)$ and $t_i \le \lv'(t_i)$
  for all component clocks $\wx$ and process clocks $t_i$. A solution
  to $\lZ \cap \lZ_2$ gives a valuation $\lv$ such that $\lv \in \lZ$
  and $\lv'$ is obtained by local elapse from $\lv$, and hence will
  imply $\lv' \in \lelapse(\lZ)$. We need to show that
  $\lZ \cap \lZ_2$ is non-empty. Note that $\lZ_2$ is a canonical
  zone. Since both $\lZ$ and $\lZ_2$ are canonical, $\lZ \cap \lZ_2$
  is empty iff there is a constraint $y_1 - y_2 \lleq c$ from $\lZ$
  and a constraint $y_2 - y_1 \lleq d$ from $\lZ_2$ which contradict
  each other (using standard literature on zones and DBMs). Inspection
  of the constraints in $\lZ_2$ and the fact that $\lv' \in \lZ_1$
  imply that this cannot happen.
\end{proof}

\begin{lemma}
  For every guard, the set $\lv$ of local valuations satisfying $g$ is
  a local zone.
\end{lemma}
\begin{proof}
  We construct a set of constraints over local variables. For every
  constraint $x\sim c$ in $g$ with $\wx\in \wX_p$, we add a constraint
  $t_p-\wx \sim c$. The obtained set of constraints gives the desired
  zone.
\end{proof}

\begin{lemma}
  For a set of clocks $R$, and a local zone $\lZ$, the set $[R]\lZ$ is
  a zone.
\end{lemma}
\begin{proof}
  For each $\wx \in R\cap \wX_p$, we remove all the edges involving
  $\wx$ in the distance graph of $\lZ$, and then add the constraints
  $\wx - t_p \leq 0$ and $t_p \leq \wx$.  The resulting zone is
  $[R]\lZ$.
\end{proof}

The operations of local time elapse, guard intersection, and reset,
enable us to describe a local step $(q,\lZ)\lact{b}(q',\lZ')$ on the
level of local zones. This is done in the same way as for global
zones.  Observe that a local step is indexed only by an action, as the
time is taken care of by the local time elapse operation.  Formally,
for an action $b$ consider a set of $b$-transitions
$\set{(q_p,g_p,R_p,q'_p)}_{p\in\dom(b)}$ of respective processes.
Then we have a transition $(q,\lZ)\act{b} (q',\lZ')$ for
$\lZ' = \lelapse([R](\lZ \cap \lZ_g \cap \lZ_{sync}))$ where
$\lZ_g = \bigcap_{p\in\dom(b)} \lZ_{g_p}$,
$\lZ_{sync} = \set{\lv \mid \lv(t_{p_1}) = \lv(t_{p_2}) \ \text{for} \
  p_1, p_2 \in \dom(b)}$ and $R = [\bigcup_{p\in\dom(b)}R_{p}]$.
Intuitively, $\lZ'$ is the set of valuations obtained through reset
and then local-time elapse, from valuations in $\lZ$ that satisfy the
guard and such that the processes involved in action $b$ are
synchronised. We extend $\lact{b}$ to $(q,\lZ)\lact{u}(q',\lZ')$ for a
sequence of actions $u$ in the obvious way.

Using local zones, we construct a local zone graph and show that it is
sound and complete for reachability testing.  The only missing step is
to verify the pre/post properties of runs on local zones.  We say a
local zone is \emph{time-elapsed} if $\lZ = \lelapse(\lZ)$.

\begin{lemma}[Pre and post properties of runs on local zones]
  \label{lem:pre-post}
  Let $u$ be a sequence of actions.
  \begin{itemize}
  \item If $(q,\lv)\lact{u}(q',\lv')$ and $\lv\in\lZ$ for some
    time-elapsed local zone $\lZ$ then $(q,\lZ)\lact{u}(q',\lZ')$ and
    $\lv'\in\lZ'$ for some local zone $\lZ'$.
  \item If $(q,\lZ)\lact{u}(q',\lZ')$ and $\lv'\in \lZ'$ then
    $(q,\lv)\lact{u}(q',\lv')$, for some $\lv\in\lZ$.
  \end{itemize}
\end{lemma}
 \begin{proof}
  This is proved by induction on the length of $u$.

  For the pre property. Suppose $u$ is a single action $a$.  Then
  $(q, \lv) \lact{a} (q', \lv')$ implies there is a sequence of local steps:
  $(q, \lv) \lact{\D} (q, \lv_1) \lact{a} (q', \lv_2) \lact{\D'} (q', \lv')$. Since $\lv \in \lZ$ and $\lZ$ is local time elapsed,
  valuation $\lv_1 \in \lZ$. By definition of $(q, \lZ) \lact{a} (q', \lZ')$ we get $\lv' \in \lZ'$. The induction step follows by a similar argument, and noting the fact that $\lZ'$
  is local-time-elapsed in $(q, \lZ) \lact{a} (q', \lZ')$.

  For the post property. When $u$ is a single action, the definition entails that there is a $\lv \in \lZ$ such that
  $(q, \lv) \lact{a} \lact{\D} (q', \lv')$. When $u = a_1 \dots a_n$, consider the sequence of zones
  $(q, \lZ) \lact{a_1} (q_1, \lZ_1) \cdots (q_{n-1}, \lZ_{n-1})
  \lact{a_n} (q', \lZ')$, use a similar argument to obtain a
  $\lv_{n-1} \in \lZ_{n-1}$ and then the induction hypothesis for the shorter sequence $a_1 \dots a_{n-1}$.
\end{proof}

\begin{definition}[Local zone graph]
  For a network of timed automata $\Nn$ the \emph{local zone graph} of
  $\Nn$, denoted $\lzg(\Nn)$, is a transition system whose nodes are
  of the form $(q,\lZ)$ where $\lZ$ is a time elapsed local zone, and
  whose transitions are steps $(q,\lZ)\lact{b}(q',\lZ')$.  The initial
  node $(q^\init, \lZ^\init)$ consists of the initial state $q_\init$
  of the network and the local zone
  $\lZ^\init = \lelapse(\{\lv^\init\})$.
\end{definition}

Directly from Lemma~\ref{lem:pre-post} we obtain the main property of
local zone graphs stated in~\cite{Bengtsson:CONCUR:1998}.  This allows
us to use local zone graphs for reachability testing.

\begin{theorem}\label{thm:offset-zone-graphs}
  For a given network of timed automata $\Nn$, there is a run of the
  network reaching a state $q$ iff for some non-empty local zone
  $\lZ$, node $(q,\lZ)$ is reachable in $\lzg(\Nn)$ from its initial
  node.
\end{theorem}

Notice that $\lzg(\Nn)$ may still be infinite and it cannot be used
directly for reachability checking. The solution in
Remark~\ref{remark:zones} does not apply to local zones due to the
multiple reference clocks. This problem will be addressed in
Section~\ref{sec:making-offset-zone}. We first focus on important
properties of the local-time zone graph w.r.t. concurrency.



\section{Why are local zone graphs better than global zone graphs?}
\label{sec:maler-zone}

The important feature about local zone graphs, as noticed
in~\cite{Bengtsson:CONCUR:1998}, is that two transitions on actions
with disjoint domains commute (see
Figure~\ref{fig:local-zones-commute}).

\begin{lemma}[Commutativity on local
  zones~\cite{Bengtsson:CONCUR:1998}]\label{lem:zone-independence}
  Suppose $\dom(a)\cap\dom(b)=\es$. If $(q,\lZ)\lact{ab} (q',\lZ')$
  then $(q,\lZ)\lact{ba}(q',\lZ')$.
\end{lemma}
\begin{proof}
  Suppose $(q,\lZ)\lact{ab} (q',\lZ'_{ab})$.  We will show that there is $(q,\lZ)\lact{ba} (q',\lZ'_{ba})$, and that
  $\lZ'_{ab}\incl\lZ'_{ba}$.  By symmetry this will show the lemma.
  
  Take $\lv'\in \lZ'_{ab}$.  Using the backward (post) property of steps on zones (Lemma~\ref{lem:pre-post}) we get a run:
  \begin{equation*}
    (q, \lv)\lact{a}(q_a, \lv_a)\lact{\D_a}(q_a, \lv'_a)\lact{b}
    (q', \lv_b)\lact{\D_b}(q', \lv') ,
  \end{equation*}
  where $\lv\in \lZ$.  Using commutation on the level of runs,
  Lemma~\ref{lem:independence}, we get a run
  \begin{equation*}
    (q, \lv)\lact{\D'_b}(q, \lv'_b)\lact{b}(q_b, \lv''_b)\lact{\D'_a}
    (q_b, \lv''_a)\lact{a}(q', \lv'''_a)\lact{\D_a''}(q', \lv') 
  \end{equation*}
  Now using the forward (pre) property of steps on zones from
  Lemma~\ref{lem:pre-post} we obtain that $\lZ_{ba}$ exists and
  $\lv'\in \lZ_{ba}$.
\end{proof}

From the above lemma, we get the following property.

\begin{corollary}\label{cor:local-zone-commutativity}
  If $(q, \lZ) \lact{u} (q', \lZ')$ and $u \sim w$, then
  $(q, \lZ) \lact{w} (q', \lZ')$
\end{corollary}

Thus starting from a local zone, all equivalent interleavings of a
sequence of actions $u$ end up in the same local zone. This is in
stark contrast to the global zone graph, where each interleaving
results in a possibly different global zone. Let
\begin{equation*}
MZ(q,Z,u)=\set{v' \mid \exists v\in Z, \ \exists w, \ w\sim u\
  \text{ and } (q,v)\act{w}(q',v')} 
\end{equation*}
denote the union of all these global zones.
Salah et al.~\cite{Salah2006OnII}\ have shown that, surprisingly,
$MZ(q,Z,u)$ is always a global zone. We call it \emph{aggregated
  zone}, and the notation $MZ$ is in the memory of Oded Maler. In the
same work, this observation was extended to an algorithm for
\emph{acyclic timed automata} that from time to time merged zones
reached by equivalent paths to a single global zone. 
We prove below that this aggregated zone can, in fact, be obtained
directly in the local zone graph: the aggregated (global) zone is
exactly the set of synchronized valuations obtained after executing
$u$ in the local zone semantics. Here we need some notation: let $Z$ be
a global zone and $\lZ$ a local zone; define
$\sync(\lZ) = \{ \lv \in \lZ \mid \lv \text{ is synchronized} \}$;
$\lval(Z) = \{\lval(v)\mid v \in Z\}$ and
$\gval(\sync(\lZ)) = \{ \gval(\lv) \mid \lv \in \sync(\lZ)
\}$. 

\begin{lemma}
  \label{lemma:sync_global_local_yields_zones}
  For every global zone $Z$ and local zone $\lZ$: $\sync(\lZ)$ and
  $\lval(Z)$ are local zones and $\gval(\sync(\lZ))$ is a global zone.
\end{lemma}
\begin{proof}
  $\sync(\lZ)$ is the local zone $\lZ \land \Land_{i,j} (t_i = t_j)$;
  $\lval(Z)$ is the local zone obtained by replacing $t$ with some
  $t_i$ in each constraint, and adding the constraints
  $\Land_{i,j} (t_i = t_j)$; $\gval(\sync(\lZ))$ is obtained by
  replacing each $t_i$ with $t$ in each constraint of $\lZ$.
\end{proof}

\begin{theorem}
  \label{thm:local_zones_are_aggregated_zones}
  Consider a state $q$, a sequence of actions $u$ and a time elapsed
  global zone $Z$. Consider the local zone $\lZ =
  \lelapse(\lval(Z))$. If $(q, \lZ) \lact{u} (q', \lZ')$, we have
  $MZ(q,Z,u)=\gval(\sync(\lZ'))$, otherwise $MZ(q,Z,u)=\es$.
\end{theorem}
\begin{proof}
  Pick $v' \in MZ(q,Z, u)$. There exists $w \sim u$, $v \in Z$ and a
  global run $(q, v) \act{w} (q', v')$. From
  Lemma~\ref{lem:eq-local-global}, there exists a local run
  $(q, \lval(v)) \lact{w} (q', \lval(v'))$. By assumption,
  $\lval(v) \in \lZ$. Hence from the pre property of local zones
  (Lemma~\ref{lem:pre-post}), there exists
  $(q, \lZ) \lact{w} (q', \lZ_w)$ such that $\lval(v') \in \lZ_w$. As
  $\lval(v')$ is synchronized, we get $\lval(v') \in
  \sync(\lZ_w)$. But, by Corollary~\ref{cor:local-zone-commutativity},
  $\lZ_w = \lZ'$. This proves $\lval(v') \in \sync(\lZ')$ and hence
  $v' \in \gval(\sync(\lZ'))$.

  For the other direction take $v'\in \gval(\sync(\lZ'))$. As $(q,
  \lZ) \lact{u} (q', \lZ')$, by
  post property of local zones (Lemma~\ref{lem:pre-post}) there is a local run $(q,
  \lv_u)\lact{u}(q',\lval(v'))$ for some $\lv_u \in \lZ$.  Since
  $\lv_u \in \lZ$, it is obtained by a local time elapse from some
  $\lv \in \lval(Z)$. Hence $\lv$ is
  synchronized and $\gval(\lv) \in Z$. From Lemma~\ref{lem:eq-local-global} we get that for some $w\sim u$ there
  is a global run $(q,\gval(\lv)) \act{w} (q',v')$. Hence $v' \in
  MZ(q, Z, u)$. 
\end{proof}

Theorem~\ref{thm:local_zones_are_aggregated_zones} gives an efficient
way to compute aggregated zones: it is sufficient to compute local zone
graphs. Computing local zone graphs is not more difficult than
computing global zone graphs. But, surprisingly, the combinatorial
explosion due to interleaving does not occur in local zone graphs,
thanks to the theorem above. Hence, this gives an incentive to work
with local zone graphs instead of global zone graphs.

This contrasts with the aggregation algorithm in~\cite{Salah2006OnII}
which requires to store all the paths to a global zone and detect
situations where zones can be merged, that is, when all the equivalent
permutations have been visited. Another important
limitation of the algorithm in~\cite{Salah2006OnII} is that it can
only be applied to acyclic zone graphs. If local zone graphs can be
computed for general timed automata (which contain cycles), we can get
to use the aggregation feature for all networks (and not only acyclic
ones). To do this, there is still a major problem left: local zone
graphs could be infinite when the automata contain cycles.


\section{Making local zone graphs finite}
\label{sec:making-offset-zone}

The standard approach to make a global zone graph finite is to use a
subsumption operation between global
zones~\cite{BBLP-STTT05,Herbreteau}. Such a subsumption operation is
usually based on a finite index simulation between (global) valuations
parameterized by certain maximum constants occurring in guards. We
first discuss technical problems that arise when we lift these
simulations to local valuations.
In the paper introducing local time semantics and local zone
graphs~\cite{Bengtsson:CONCUR:1998} a notion of a \emph{catch-up
  equivalence} between local valuations is defined. This equivalence
is a finite index simulation.  So one could, in theory, construct a
finite local zone graph using catch-up subsumption as a finite
abstraction.  Unfortunately, the question of effective algorithms for
catch-up subsumption was left open in~\cite{Bengtsson:CONCUR:1998},
and, to the best of our knowledge, there is no efficient procedure for
catch-up subsumption.

Another finite abstraction of the local zone graph was proposed by
Minea~\cite{Minea,Minea:thesis}.  We however believe that Minea's
approach carries a flaw, and a different idea is needed to get
finiteness.  Minea's approach is founded on an equivalence between
local valuations in the lines of the region
equivalence~\cite{Alur:TCS:1994}. Let $\cmax$ be the maximum constant
appearing among the guards of a timed automata network. Two local
valuations $\lv_1$ and $\lv_2$ are said to be equivalent, written as
$\lv_1 \mineareg \lv_2$ if for all variables
$\a, \b \in \wX \cup \{t_1,\dots,t_k\}$ (note that all reference
clocks are included):
\begin{itemize}
\item either $\lfloor \lv_1(\a) - \lv_1(\b) \rfloor$ $=$
  $\lfloor \lv_2(\a) - \lv_2(\b) \rfloor$,
\item or $\lfloor \lv_1(\a) - \lv_1(\b) \rfloor > \cmax$ and
  $\lfloor \lv_2(\a) - \lv_2(\b) \rfloor > \cmax$,
\item or $\lfloor \lv_1(\a) - \lv_1(\b) \rfloor < -\cmax$ and
  $\lfloor \lv_2(\a) - \lv_2(\b) \rfloor < -\cmax$.
\end{itemize}

It is claimed (in Proposition 6 of \cite{Minea}) that this equivalence
is preserved over local time elapse: for every process $p$, and for
$\d \ge 0$ there exists $\d' \ge 0$ such that
$\lv_1 +_p \d \mineareg \lv_2 +_p \d'$.  However, this is not true, as
we now exhibit a counter-example. Consider $2$ processes with clocks
$X_1 = \{x\}$, $X_2 = \{y\}$. This gives $\wX = \{ \wx, \wy \}$ and
$T = \{t_1, t_2\}$. Let $\cmax = 3$. Define local valuations $\lv_1:
\wx = 0, t_1 = 0, \wy = 0, t_2 = 4$ and $\lv_2: \wx = 0, t_1 = 0, \wy
= 0, t_2 = 5$.

Note that the differences in $\lv_1$ are either $0, 4$ or $-4$ and the
corresponding differences in $\lv_2$ are $0, 5$ or $-5$. Hence by
definition, $\lv_1 \mineareg \lv_2$. Consider valuation $\lv_1 +_1 2$
obtained from $\lv_1$ by local delay of $2$ units in component $1$,
that is
  $\lv_1 +_1 2: \wx = 0, t_1 = 2, \wy = 0, t_2 = 4$
Observe that in $\lv_1 +_1 2$, the difference $t_1 - \wx = 2$ and
$t_2 - t_1 = 2$ which are both smaller than $\cmax$. We claim there is
no local delay $\d'$ such that $\lv_1 +_1 2 \mineareg \lv_2 +_1
\d'$. Valuation $\lv_2 +_1 \d'$ is given by
$\wx = 0, t_1 = \d', \wy = 0, t_2 = 5$. If
$\lv_1 +_1 2 \mineareg \lv_2 +_1 \d'$, we need $\d' = 2$ (due to
difference $t_1 - \wx$) and $5 - \d' = 2$ (due to difference
$t_2 - t_1$). This is not possible.

The main problem is that the above equivalence ``forgets'' actual
values when the difference between reference clocks is above
$c_{max}$. Even if this difference is bigger than the maximum
constant, local delays can bring them within the constant
$c_{max}$. Such a situation does not arise in the global semantics, as
there is a single reference clock.

In \cite{Minea} a widening operator on local zones based on
$\mineareg$ is used for finiteness: given a canonical representation
of a zone $\lZ$, the \emph{maximized zone} with respect to $\cmax$ is
obtained by changing every constraint $y_1 - y_2 \lleq c$ to
$y_1 - y_2 < \infty$ if $c > \cmax$, and to $y_1 - y_2 < - \cmax$ if
$c < -\cmax$. In the local zone graph construction, each local zone is
maximized and inclusion between maximized zones is used for
termination. Figure~\ref{fig:minea-bug} gives an example of a network
$\langle A_1, A_2 \rangle$ where the maximized local zone graph is
unsound. This is shown by making use of the valuations $\lv_1$ and
$\lv_2$ above. We also add an extra clock  $z$ in
component $A_2$ for convenience. 
Note that $\cmax = 3$. Although clock $y$ does not
appear in $A_2$, one can assume that there are other transitions from
$q_0$ that deal with $y$  (we avoid illustrating these
transitions explicitly).
In the discussion below, $\lv_1, \lv_2$ are valuations
restricted to $\wx, \wy, t_1$ and $t_2$.  In order to reach the state
$p_2$, the synchronization action $c$ needs to be taken: transition
sequence $a_1 c$ requires $c$ to be taken at global time $4$, and
transition sequence $b_1 b_2 c$ requires $c$ at global time $5$. Hence
$c$ is not enabled in the network. This is witnessed by $c$ not being
enabled in the local zone graph (middle picture). Valuation $\lv_2$ is
present in the zone reached after $b_1 b_2$. The maximized local zone
graph is shown on the right.  Zones where maximization makes a
difference are shaded gray. In particular, the zone $b_1 b_2$ on
maximization adds valuation $\lv_1$, from which $a_1 c$ is enabled,
giving a zone in the maximized local zone graph with state $p_2$.

\begin{figure}[t]
  \centering
  \begin{tikzpicture}[state/.style={draw, thick, circle, inner
      sep=2pt}]
    \begin{scope}[every node/.style={state}]
      \node (p0) at (0, 10) {\scriptsize $p_0$};
      \node (p1) at (0, 8.5) {\scriptsize $p_1$};
      \node (p2) at (0, 7) {\scriptsize $p_2$};
      
      \node (q0) at (1, 10) {\scriptsize $q_0$};
      \node (q1) at (1, 8.5) {\scriptsize $q_1$};
      \node (q2) at (1, 7) {\scriptsize $q_2$};
      \node (q3) at (1, 5.5) {\scriptsize $q_3$};
    \end{scope}
    
    \begin{scope}[->, >=stealth, thick]
      \draw (0, 10.8) to (p0);
      \draw (p0) to (p1);
      \draw (p1) to (p2);
      \draw (1, 10.8) to (q0);
      \draw (q0) to (q1);
      \draw (q1) to (q2);
      \draw (q2) to (q3);
    \end{scope}
    
    \node at (0.2, 9.25) {\footnotesize $a_1$};
    \node at (0.2, 7.75) {\footnotesize $c$};
    \node at (0.8, 9.25) {\footnotesize $b_1$};
    \node at (0.8, 7.75) {\footnotesize $b_2$};
    \node at (0.8, 6.25) {\footnotesize $c$};
    
    \node [left] at (0, 9.35) {\footnotesize $x = 2$};
    \node [left] at (-0.1, 9) {\footnotesize $\{x\}$};
    \node [left] at (0, 7.75) {\footnotesize $x = 2$};
    \node [right] at (1, 9.35) {\footnotesize $z = 2$};
    \node [right] at (1.1, 9) {\footnotesize $\{z\}$};
    \node [right] at (1, 7.85) {\footnotesize $z = 3$};
    \node [right] at (1.1, 7.5) {\footnotesize $\{z \}$};
    
    \draw [gray, thick] (1.9, 10.8) -- (1.9, 5);
    \draw [gray, thick] (6.9, 10.8) -- (6.9, 5);

    \begin{scope}[xshift=2.7cm, yshift=-0.3cm,
      zone/.style={draw,rectangle,rounded corners,inner
        sep=1pt},font=\scriptsize]

      \node [zone] (z0) at (2, 10.5)
      {$\begin{array}{l}
          \wx=\wy=\wz\\
          t_1\ge \wx,\, t_2\ge \wz
        \end{array}$};
      \node [zone] (z1) at (0.5, 9)
      {$\begin{array}{l}
          \wy=\wz\\
          \wx-\wy=2\\
          t_1\ge \wx,\, t_2\ge \wz
        \end{array}$};
      \node [zone] (z2) at (3, 9)
      {$\begin{array}{l}
          \wx=\wy\\
          \wz-\wy=2\\
          t_1\ge \wx,\, t_2\ge \wz
        \end{array}$};
      \node [zone] (z3) at (0.5, 7.5)
      {$\begin{array}{l}
          \wx-\wy=2\\
          \wz-\wy=2\\
          t_1\ge \wx,\, t_2\ge \wz
        \end{array}$};
      \node [zone] (z4) at (3, 7.5)
      {$\begin{array}{l}
          \wx=\wy\\
          \wz-\wy=5\\
          t_1\ge \wx,\, t_2\ge \wz
        \end{array}$};
      \node [zone] (z5) at (0.5, 6)
      {$\begin{array}{l}
          \wx-\wy=2\\
          \wz-\wy=5\\
          t_1\ge \wx,\, t_2\ge \wz
        \end{array}$};

      \node at (2, 11) {$\langle p_0, q_0 \rangle$};
      \node at (0, 9.7) {$\langle p_1, q_0 \rangle$};
      \node at (3.5, 9.7) {$\langle p_0, q_1 \rangle$};
      \node at (0, 8.2) {$\langle p_1, q_1 \rangle$};
      \node at (3.5, 8.2) {$\langle p_0, q_2 \rangle$};
      \node at (0, 6.7) {$\langle p_1, q_2 \rangle$};
      
      \begin{scope}[->, >=stealth, gray, thick]
        \draw (z0) to node [midway, black, fill=white, inner sep=0.5pt]
        {\scriptsize $a_1$} (z1);
        \draw (z0) to node [midway, black, fill=white, inner
        sep=0.5pt]
        {\scriptsize $b_1$} (z2);
        \draw (1.3, 8.5) to node [midway, black, fill=white, inner
        sep=0.5pt]
        {\scriptsize $b_1$} (1.3,8);
        \draw (z2) to node [midway, black, fill=white, inner
        sep=0.5pt]
        {\scriptsize $a_1$} (z3);
        \draw (2.7, 8.5) to node [midway, black, fill=white, inner
        sep=0.5pt]
        {\scriptsize $b_2$} (2.7, 8);
        \draw (1.3, 7) to node [midway, black, fill=white, inner
        sep=0.5pt]
        {\scriptsize $b_2$} (1.3, 6.5);
        \draw (z4) to node [midway, black, fill=white, inner
        sep=0.5pt]
        {\scriptsize $a_1$} (z5);
      \end{scope}
    \end{scope}

    \begin{scope}[xshift=7.7cm, yshift=-0.3cm,
      zone/.style={draw,rectangle, rounded corners,inner
        sep=1pt},font=\scriptsize]

      \node [zone] (z0) at (2, 10.5)
      {$\begin{array}{l}
          \wx=\wy=\wz\\
          t_1\ge \wx,\, t_2\ge \wz
        \end{array}$};
      \node [zone] (z1) at (0.5, 9)
      {$\begin{array}{l}
          \wy=\wz\\
          \wx-\wy=2\\
          t_1\ge \wx,\, t_2\ge \wz
        \end{array}$};
      \node [zone] (z2) at (3, 9)
      {$\begin{array}{l}
          \wx=\wy\\
          \wz-\wy=2\\
          t_1\ge \wx,\, t_2\ge \wz
        \end{array}$};
      \node [zone] (z3) at (0.5, 7.5)
      {$\begin{array}{l}
          \wx-\wy=2\\
          \wz-\wy=2\\
          t_1\ge \wx,\, t_2\ge \wz
        \end{array}$};
      \node [zone,fill=gray!20] (z4) at (3, 7.5)
      {$\begin{array}{l}
          \wx=\wy\\
          \wz-\wy=3\\
          t_1\ge \wx,\, t_2\ge \wz
        \end{array}$};
      \node [zone,fill=gray!20] (z5) at (0.5, 6)
      {$\begin{array}{l}
          \wx-\wy=2\\
          \wz-\wy=3\\
          t_1\ge \wx,\, t_2\ge \wz
        \end{array}$};
      \node [zone,fill=green!20] (z6) at (3.5,6)
      {$\begin{array}{l}
          \wx-\wy=2, \, \wz-\wy=3\\
          t_1-\wx \ge 2\\
          t_2-\wy \ge 3, \, t_2 \ge \wz
        \end{array}$};

      \node at (2, 11) {$\langle p_0, q_0 \rangle$};
      \node at (0, 9.7) {$\langle p_1, q_0 \rangle$};
      \node at (3.5, 9.7) {$\langle p_0, q_1 \rangle$};
      \node at (0, 8.2) {$\langle p_1, q_1 \rangle$};
      \node at (3.5, 8.2) {$\langle p_0, q_2 \rangle$};
      \node at (0, 6.7) {$\langle p_1, q_2 \rangle$};
      \node at (4, 6.7) {$\langle p_2, q_3 \rangle$};
      
      \begin{scope}[->, >=stealth, gray, thick]
        \draw (z0) to node [midway, black, fill=white, inner sep=0.5pt]
        {\scriptsize $a_1$} (z1);
        \draw (z0) to node [midway, black, fill=white, inner
        sep=0.5pt]
        {\scriptsize $b_1$} (z2);
        \draw (1.3, 8.5) to node [midway, black, fill=white, inner
        sep=0.5pt]
        {\scriptsize $b_1$} (1.3,8);
        \draw (z2) to node [midway, black, fill=white, inner
        sep=0.5pt]
        {\scriptsize $a_1$} (z3);
        \draw (2.7, 8.5) to node [midway, black, fill=white, inner
        sep=0.5pt]
        {\scriptsize $b_2$} (2.7, 8);
        \draw (1.3, 7) to node [midway, black, fill=white, inner
        sep=0.5pt]
        {\scriptsize $b_2$} (1.3, 6.5);
        \draw (z4) to node [midway, black, fill=white, inner
        sep=0.5pt]
        {\scriptsize $a_1$} (z5);
        \draw (z5) to node [midway, black, fill = white, inner
        sep=0.5pt]
        {\scriptsize $c$} (z6);
      \end{scope}
    \end{scope}    
  \end{tikzpicture}
  
  \caption{\emph{Left:} network $\langle A_1, A_2 \rangle$;
    \emph{Middle:} local zone graph; \emph{Right:} maximized local
    zone graph of~\cite{Minea}. State $(p_2, q_3)$ is not reachable in
    local zone graph, but becomes reachable after maximization.}
  \label{fig:minea-bug}
\end{figure}
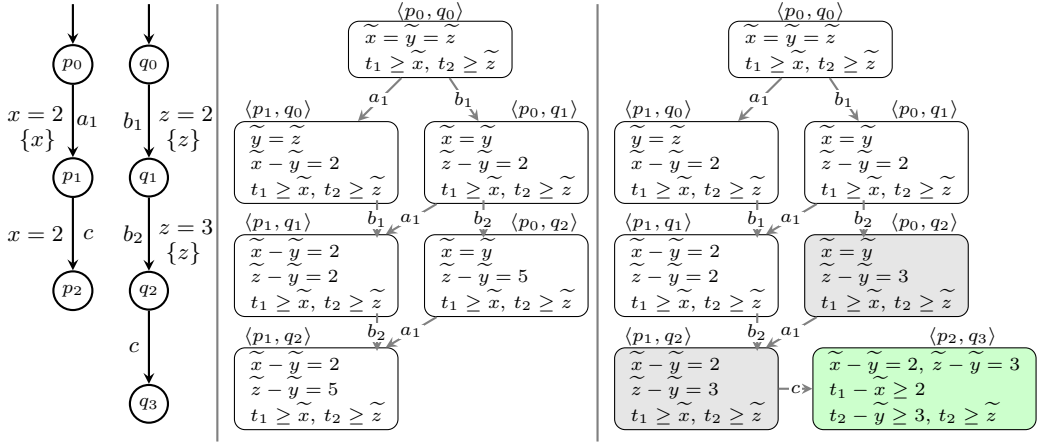


\subsection{Synchronized valuations for subsumption}

A finite abstraction of the differences between reference clocks
constitutes the main challenge in obtaining a finite local zone
graph. We propose a different solution which bypasses the need to
worry about such differences: \emph{restrict to synchronized
  valuations for subsumption}.

Subsumptions over global zones are well
studied~\cite{Alur:TCS:1994,BBLP-STTT05,Herbreteau}.  Taking an
off-the-shelf finite abstraction $\abs$ and a subsumption
$\sqsubseteq^\abs$ between global zones, we want to do the following:
given two local zones $\lZ_1$ and $\lZ_2$ we perform a subsumption
test $\gval(\sync(\lZ_1)) \sqsubseteq^\abs \gval(\sync(\lZ_2))$.  By
finiteness of the abstraction, we will get only finitely many local
zones $\lZ$ with incomparable $\gval(\sync(\lZ))$. In order to do
this, we need the crucial fact that $\gval(\sync(\lZ))$ is a zone
(shown in Lemma~\ref{lemma:sync_global_local_yields_zones}).
Given two configurations $s:= (q, \lZ)$ and $s':= (q', \lZ')$ of the
local zone graph, we write $s \syncincl s'$ if $q = q'$ and
$\gval(\sync(\lZ)) \sqsubseteq^\abs \gval(\sync(\lZ'))$.

\begin{definition}[Local sync graph] \label{def:sync-graph} Let $\abs$
  be a finite abstraction over global zones. A \emph{local sync graph}
  $G$ of a network of timed automata $\Nn$ based on $\abs$, is a tree
  and a subgraph of $\lzg(\Nn)$ satisfying the following conditions:
  \begin{description}
  \item[\textbf{C0}] every node of $G$ is labeled either
    \emph{covered} or \emph{uncovered};
  \item[\textbf{C1}] the initial node of $\lzg(\Nn)$ belongs to $G$
    and is labeled uncovered;
  \item[\textbf{C2}] every node is reachable from the initial node;
  \item [\textbf{C3}] for every uncovered node $s$, all its successor
    transitions $s \lact{a} s'$ occurring in $\lzg(\Nn)$ should be
    present in $G$;
  \item [\textbf{C4}] for every covered node $s \in G$ there is an
    uncovered node $s' \in G$ such that $s \syncincl s'$. A covered
    node has no successors.
  \end{description}
\end{definition}

The above definition essentially translates to this algorithm: explore
the local zone graph say in a BFS fashion, and subsume (cover) using $\syncincl$ (similar
to algorithm for global zone graph as described in
page~\pageref{def:zone-graph-subsumption}).  Local sync graphs are not
unique, since the final graph depends on the order of exploration.
Every local sync graph based on a finite abstraction $\abs$ is finite.
Theorem~\ref{thm:main} below states that local sync graphs are
sound and complete for reachability, and this algorithm is correct.
\begin{theorem}[Soundness and completeness of local sync
  graphs]\label{thm:main}
  A state $q$ is reachable in a network of timed automata $\Nn$ iff a
  node $(q,\lZ)$, with $\lZ$ non-empty, is reachable from the initial
  node in a local sync graph for $\Nn$.
\end{theorem}
\begin{proof}
  If $(q,\lZ)$ is reachable in a local sync graph then it is trivially
  reachable in the local zone graph and the (backward) implication
  follows from soundness of local zone graphs (Theorem
  \ref{thm:offset-zone-graphs}). For the other direction which proves
  completeness of local sync graphs, let us take a global run:
  $(q_\init,v_\init)\act{\d_1,b_1}(q_1,v_1)\cdots\act{\d_n,b_n}(q_n,v_n)$. Take
  a local sync graph $G$, based on abstraction $\abs$ coming from a
  simulation $\fleq$.  By induction on $i$, for every $(q_i,v_i)$ we
  will find an uncovered node $(q_i,\lZ_i)$ of $G$, and a synchronized
  local valuation $\lv_i\in \lZ_i$ such that $v_i \fleq
  \gval(\lv_i)$. This proves completeness, since every reachable
  $(q_i, v_i)$
  will have a representative node in the local sync graph.
  
  The induction base is immediate, so let us look at the induction
  step. Consider the global step
  $(q_i, v_i) \act{\d_i, b_i} (q_{i+1}, v_{i+1})$. Since
  $v_i \fleq \gval(\lv_i)$, there is a delay $\d'_i$ such that
  $(q_i, \gval(\lv_i)) \act{\d'_i, b_i} (q_{i+1}, v'_{i+1})$ and
  $v_{i+1} \fleq v'_{i+1}$. As the global delay $\d'_i$ can be thought
  of a sequence of local delays, we have local run
  $(q_i, \lv_i) \lact{b_i} (q_{i+1}, \lv'_{i+1})$, where
  $\lv'_{i+1} = \lval(v'_{i+1})$. Note that $\lv'_{i+1}$ is
  synchronized and $v_{i+1} \fleq \gval(\lv'_{i+1})$. From the
  pre-property of local zones (Lemma~\ref{lem:pre-post}) there exists
  a transition $(q_i, \lZ_i) \lact{b_i} (q_{i+1}, \lZ'_{i+1})$ with
  $\lv'_{i+1} \in \lZ'_{i+1}$, in fact,
  $\lv'_{i+1} \in \sync(\lZ'_{i+1})$.  If $(q_{i+1}, \lZ'_{i+1})$ is
  uncovered, take $\lv'_{i+1}$ for $\lv_{i+1}$ and $\lZ'_{i+1}$ for
  $\lZ_{i+1}$ (needed by the induction step). Otherwise, from
  condition C4, there is an uncovered node $(q_{i+1}, \lZ''_{i+1})$
  such that
  $\gval(\sync(\lZ'_{i+1})) \fleq \gval(\sync(\lZ''_{i+1}))$.  This
  gives $\lv''_{i+1} \in \sync(\lZ''_{i+1})$ such that
  $\gval(\lv'_{i+1}) \fleq \gval(\lv''_{i+1})$.  Now take
  $\lv''_{i+1}$ for $\lv_{i+1}$ and $\lZ''_{i+1}$ for $\lZ_{i+1}$.
\end{proof}


\section{Experiments}
\label{sec:experiments}

We have implemented the construction of local sync graphs in our
prototype TChecker~\cite{TChecker} and compared it with two
implementations of the usual global zone graph method: TChecker and
UPPAAL~\cite{Larsen:1997:UPPAAL,BehrmannDLHPYH06}, the state-of-the-art
verification tool for timed automata. 
The three implementations use a breadth-first
search  with subsumption, and the $\syncincl$ subsumption in the case
of local sync graph.
Table~\ref{table:experimental-results} presents results of our
experiments on standard models from the literature (except
``Parallel'' that is a model we have introduced). 

Local sync graphs yield no gain on 3 standard examples (which are not
given in Table~\ref{table:experimental-results}): ``CSMA/CD'',
``FDDI'' and ``Fischer''. In these models, the three algorithms visit
and store the 
same number of nodes.
The reason is that for ``CSMA/CD'' and ``FDDI'', replacing
each local zone $\lZ$ in the local zone graph by its set of
synchronized valuations $sync(\lZ)$ yields exactly the zone graph. In
the third model, ``Fischer'', every control state appears at most once
in the global zone graph. So there is no hope to achieve any gain with
our technique. This is due to the fact that doing $ab$ or $ba$ results
in two different control states in the automaton. So ``Fischer'' is
out of the scope of our technique.

In contrast, we observe significant improvements on other standard
models (Table~\ref{table:experimental-results}).
Observe that due to
subsumption, the order in which nodes are visited impacts the total
number of visited nodes. UPPAAL and our prototype TChecker (Global ZG
column) may not visit the same number of nodes despite the fact that
they implement the same algorithm.  In our prototype we use the same
order of exploration for Global ZG and Local ZG. ``CorSSO'' and
``Critical region'' are standard examples from the
literature. ``Dining Philosophers'' is a modification of the classical
problem where a philosopher releases her left fork if she cannot take
her right fork within a fixed amount of
time~\cite{DBLP:journals/tcs/LugiezNZ05}.  ``Parallel'' is a model we
have introduced, where concurrent processes compete to access a
resource in mutual exclusion. We observe an order of magnitude gains
for most of these four models.
The reason is that in most states
when two processes can perform actions $a$ and $b$, doing $ab$ or $ba$
leads to the same control-state of the automaton. Hence, a difference
between $ab$ and $ba$ (if any) is encoded in distinct zones $Z_{ab}$
and $Z_{ba}$ obtained along these two paths in the global zone
graph. In contrast, the two paths result in the same zone $\lZ$
(containing both $Z_{ab}$ and $Z_{ba}$) in the local sync graph. In
consequence, our approach that combines the local zone graph and
abstraction using synchronized zones is very efficient in this
situation.

\begin{table}[t]
  \centering
  \begin{tabular}{|l|r|r|r|r|r|r|r|r|r|}
    \hline
    Models
    & \multicolumn{3}{|c|}{UPPAAL}
    & \multicolumn{3}{|c|}{Global ZG}
    & \multicolumn{3}{|c|}{Local ZG}\\
    \cline{2-10}
    (\# processes)
    & visited & stored & sec.
    & visited & stored & sec.
    & visited & stored & sec.\\
    \hline
    CorSSO 3
    & 64378 & 61948 & 1.48
    & 64378 & 61948 & 1.41
    & 1962 & 1962 & 0.05\\
    CorSSO 4
    & \multicolumn{3}{|c|}{timeout}
    & \multicolumn{3}{|c|}{timeout}
    & 23784 & 23784 & 0.69\\
    CorSSO 5
    & \multicolumn{3}{|c|}{timeout}
    & \multicolumn{3}{|c|}{timeout}
    & 281982 & 281982 & 16.71\\
    \hline
    Critical reg. 4
    & 78049 & 53697 & 1.45
    & 75804 & 53697 & 2.27
    & 44490 & 28400 & 2.40\\
    Critical reg. 5
    & \multicolumn{3}{|c|}{timeout}
    & \multicolumn{3}{|c|}{timeout}
    & 709908 & 389614 & 75.55\\
    \hline
    Dining Phi. 7
    & 38179 & 38179 & 34.61
    & 38179 & 38179 & 7.28
    & 2627 & 2627 & 0.32\\
    Dining Phi. 8
    & \multicolumn{3}{|c|}{timeout}
    & \multicolumn{3}{|c|}{timeout}
    & 8090 & 8090 & 1.65\\
    Dining Phi. 9
    & \multicolumn{3}{|c|}{timeout}
    & \multicolumn{3}{|c|}{timeout}
    & 24914 & 24914 & 7.10\\
    Dining Phi. 10
    & \multicolumn{3}{|c|}{timeout}
    & \multicolumn{3}{|c|}{timeout}
    & 76725 & 76725 & 30.20\\
    \hline
    Parallel 6
    & 11743 & 11743 & 4.82
    & 11743 & 11743 & 1.09
    & 256 & 256 & 0.02\\
    Parallel 7
    & \multicolumn{3}{|c|}{timeout}
    & \multicolumn{3}{|c|}{timeout}
    & 576 & 576 & 0.04\\
    Parallel 8
    & \multicolumn{3}{|c|}{timeout}
    & \multicolumn{3}{|c|}{timeout}
    & 1280 & 1280 & 0.11\\
    \hline
  \end{tabular}
  \caption{Experimental results obtained by running UPPAAL and our
    prototype TChecker (Global ZG and Local ZG) on a MacBook Pro 2013
    with 4 2.4GHz Intel Core i5 and 16 GB of memory. The timeout is 90
    seconds. For each model we report the number of concurrent
    processes.}
  \label{table:experimental-results}
\end{table}


\section{Conclusions}
\label{sec:conclusions}
We have revisited local-time semantics of timed automata and local
zone graphs. We have discovered a very useful fact that local zones
calculate aggregated zones: global zones that are unions of all the
zones obtained by equivalent executions~\cite{Salah2006OnII}.  We have
used this fact as a theoretical foundation for an algorithm
constructing local zone graphs and using subsumption on aggregated
zones at the same time.

We have shown that, unfortunately, subsumption operations on local zones
proposed in the literature do not work. We have proposed a new
subsumption for local zone graphs based on standard abstractions for
timed automata, applied on synchronized zones. The restriction to
synchronized zones is crucial as standard abstractions 
cannot handle multiple reference clocks. A direction for future work is
to find  abstractions for local zones.

Our algorithm is the first efficient implementation of local time zone
graphs and aggregated zones. Experimental results show an order of
magnitude gain with respect to state-of-the-art algorithms on several
standard examples.

As future work, we plan to develop partial-order techniques taking
advantage of the high level of commutativity in local zone graphs.
Existing methods are not directly
applicable in the timed setting. 
In particular, contrary to expectations, actions with disjoint domains
may not be independent (in a partial order sense) in a local zone
graph~\cite{Minea}.
 Thus, it will be very
interesting to understand the structure of local zone graphs better. A
recent partial-order method proposed for timed-arc Petri
nets~\cite{DBLP:conf/cav/BonnelandJLMS18} gives a hope that such
obstacles can be overcome.  For timed networks with cycles, the
interplay of partial-order and subsumption adds another level of
difficulty.
  

\bibliography{por}

\end{document}